\documentclass[journal,twoside,web]{ieeecolor}
\usepackage{lcsys}
\usepackage{cite}
\usepackage{amsmath,amssymb,amsfonts}
\usepackage{algorithmic}
\usepackage{graphicx}
\usepackage{textcomp}
\def\BibTeX{{\rm B\kern-.05em{\sc i\kern-.025em b}\kern-.08em
    T\kern-.1667em\lower.7ex\hbox{E}\kern-.125emX}}

\usepackage[utf8]{inputenc}
\usepackage{balance,soul}
\usepackage{color}
\usepackage[caption=false]{subfig}

\graphicspath{{./Figures/}}

\newcommand{\Real}{\mathbb{R}}
\newcommand{\Nat}{\mathbb{N}}

\newcommand{\Acal}{\ensuremath{\mathcal{A}}}
\newcommand{\Bcal}{\ensuremath{\mathcal{B}}}

\newcommand{\Mcal}{\ensuremath{\mathcal{M}}}
\newcommand{\Ecal}{\ensuremath{\mathcal{E}}}

\newcommand{\Fcal}{\ensuremath{\mathcal{F}}}
\newcommand{\Gcal}{\ensuremath{\mathcal{G}}}
\newcommand{\Scal}{\ensuremath{\mathcal{S}}}
\newcommand{\Vcal}{\ensuremath{\mathcal{V}}}
\newcommand{\Pcal}{\ensuremath{\mathcal{P}}}
\newcommand{\Rcal}{\ensuremath{\mathcal{R}}}
\newcommand{\Ncal}{\ensuremath{\mathcal{N}}}

\newcommand{\Ucal}{\ensuremath{\mathcal{U}}}

\newcommand{\Tcal}{\ensuremath{\mathcal{T}}}
\newcommand{\Zcal}{\ensuremath{\mathcal{Z}}}
\newcommand{\fM}{\mathfrak{M}}
\newcommand{\bS}{\mathbb{S}}

\newcommand{\stkout}[1]{\ifmmode\text{\sout{\ensuremath{#1}}}\else\sout{#1}\fi}

\newcommand{\varempty}{\emptyset}

\newtheorem{defn}{Definition}
\newtheorem{lemma}{Lemma}
\newtheorem{theorem}{Theorem}
\newtheorem{assumption}{Assumption}

\title{Excitation allocation for generic identifiability of linear dynamic networks with fixed modules}
\author{H.J.~Dreef, S.~Shi, X.~Cheng, M.C.F.~Donkers and P.M.J.~Van~den~Hof \thanks{This project has received funding from the European Research Council
		(ERC), Advanced Research Grant SYSDYNET, under the European
		Union’s Horizon 2020 research and innovation programme (grant agreement No. 694504).} \thanks{Mannes~Dreef, Tijs~Donkers and Paul~Van~den~Hof are with the Dept. of Electrical Engineering, Eindhoven University of Technology, Netherlands. Shengling Shi is  with the Delft Center for Systems and Control, Delft University of Technology, The Netherlands. Xiaodong Cheng is with the Department of Engineering, University of Cambridge, UK. {Email: \tt \{h.j.dreef,m.c.f.donkers,p.m.j.vandenhof\}@tue.nl}.}}

\begin{document}
\maketitle
\thispagestyle{empty}

\begin{abstract}
Identifiability of linear dynamic networks requires the presence of a sufficient number of external excitation signals. The problem of allocating a minimal number of external signals for guaranteeing generic network identifiability in the full measurement case has been recently addressed in the literature. Here we will extend that work by explicitly incorporating the situation that some network modules are known, and thus are fixed in the parametrized model set. The graphical approach introduced earlier is extended to this situation, showing that the presence of fixed modules reduces the required number of external signals. An algorithm is presented that allocates the external signals in a systematic fashion.
\end{abstract}

\section{Introduction}
\IEEEPARstart{T}{he} recent attention for dynamical systems in the format of structured interconnections of individual dynamical subsystems, has generated many challenging research questions. In the area of data-driven modeling, attention has been given to methods for modeling either a complete network or a particular subsystem on the basis of (a selection of) measured network signals. In these problems the topology of the network, i.e. the interconnection structure, can either be known or also be subject of identification  \cite{Goncalves&Warnick:08,Materassi&Innocenti:10,Nabi&Mesbahi:12,VandenHof&etal_Autom:13,Materassi&Salapaka:20,Ramaswamy&VandenHof_TAC:21}.

For the problem of identifying a full linear dynamic network, in the situation that the topology is known, one of the key questions is whether there is enough excitation present in the data for uniquely recovering the dynamics of the network. This question is phrased through the concept of network identifiability \cite{Weerts&etal_Autom:18_identifiability,Hendrickx&Gevers&Bazanella_TAC:19,vanWaarde&etal_TAC:20a,Shi&etal_Autom:22}, being essentially dependent on the external signals present in the network, involving both user-chosen excitation signals and unmeasured disturbance signals, as well as on the structure of the model set reflecting possible prior knowledge on the network dynamics.

Recently, an algorithm has been presented in \cite{Cheng&etal_TAC:22} that aims to allocate a minimum number of external excitation signals in the network, so as to guarantee that a network model set becomes generically identifiable, for the situation that all nodes in the network are measured (full measurement). By focusing on {\it generic} identifability, i.e. the property holding for {\it almost all} models in the considered model set, the conditions for identifiability can be formulated in terms of the graph of the underlying model set, and are thus becoming easily applicable.

The graphical algorithm presented in \cite{Cheng&etal_TAC:22} decomposes the network graph into a set of distinct pseudotrees, on the basis of which allocation of external excitation signals can simply be executed.
In many network systems, the dynamics of some particular modules may be known and fixed, as e.g., user-designed controllers. The algorithm in \cite{Cheng&etal_TAC:22} can accommodate this to some extent by requiring that fixed modules do not need to be covered by the pseudotrees, possibly leading to a smaller number of required excitation signals. However it appears that this approach can be conservative, in the sense that it does not fully exploit the benefit of having some modules fixed, and therefore may yield more excitation signals than necessary. Therefore the key research question in this paper is: can we take advantage of the prior knowledge of these fixed modules to allocate excitation signals in a more efficient way?

In this paper, we extend the framework and algorithm in \cite{Cheng&etal_TAC:22} to explicitly incorporate fixed modules, by relaxing and generalizing
the graphical concept of pseudotree to a new concept (single-source identifiable multi-rooted graph (SIMUG)) that will be used to cover the graph of the network model set, and will be shown to provide the means for allocating external excitation signals more effectively.

After defining the appropriate network concepts in Section \ref{sec:prelim}, a motivating example and the resulting problem statement will be presented in Section \ref{sec:motivexam}. This leads to a new result for identifiability and a related allocation algorithm presented in Sections \ref{sec:approach}-\ref{sec:allo}. Finally a brief example is provided.

\textit{Nomenclature:} Denote $\Nat$ and $\Real$ as the sets of natural and real numbers; $\Real(z)$ is the rational function field over $\Real$ with variable $z$. $v_i$ denotes the $i$-th element of a vector $v$, and $A_{ij}$ denotes the $(i,j)$-th entry of a matrix $A$, and $A_{\star i}$ ($A_{i\star}$) its $i$-th column (row). The cardinality of a set $\Vcal$ is given by $|\Vcal|$. The edges of directed graph $\Gcal$ are denoted by $E(\Gcal)$ and its vertices by $V(\Gcal)$.

\section{Preliminaries}
\label{sec:prelim}

\subsection{Dynamic network setup}
We consider a dynamic network following the setup in \cite{VandenHof&etal_Autom:13}, which describes the dynamics and interconnection between a set of measured internal nodes
$\{w_1, \dots,w_L\}$.
The set of measured external excitation signals $\{ r_1, \dots, r_K \}$ can be manipulated by the user and the set of unmeasured disturbance signals is given by $\{v_1,\dots,v_L\}$. The expression for each node is given by
\begin{equation} \label{eq:network_setup}
	w_j(t) = \sum_{i=1}^{L} G_{ji}(q) w_i(t) + \sum_{k=1}^{K} R_{jk}(q) r_k(t) + v_j(t)
\end{equation}
where $G_{ji}(q), R_{jk}(q) \in \Real(q)$ are rational transfer functions that connect the nodes and excitation signals, with time-shift operator $q$ such that $q^{-1} w_j(t) = w_j(t-1)$. The elements $G_{ji}(q)$ are called modules, where $G_{jj}(q) = 0$, to exclude self-loops. The unmeasured process noise variables $v_j$ are collected in the vector process ${v = [v_1 \dots v_L]^\top}$, which is modeled as a stationary stochastic process with rational spectral density $\Phi_{v}(\omega)$, such that there exists a $p$-dimensional (zero-mean) white noise process ${e:=[e_1 \dots e_p]^\top}$, with ${p\leq L}$ and covariance matrix $\Lambda > 0$ such that
	$v(t) = H(q)e(t)$.
The combination of all the $L$ nodes can be written in terms of the full network expression
{\small \[ \setlength\arraycolsep{1pt} 
		\begin{bmatrix}
			w_1 \\ w_2 \\ \vdots \\ w_L
		\end{bmatrix} \! = \! \begin{bmatrix}
			0 & G_{12} & \cdots & G_{1 L} \\ G_{2 1} & 0 & \ddots & \vdots \\ \vdots & \ddots & \ddots & G_{L-1 L} \\ G_{L 1}& \cdots & G_{L L-1} &  0
		\end{bmatrix} \! \! \begin{bmatrix}
			w_1 \\ w_2 \\ \vdots \\ w_L
		\end{bmatrix} \!   + \! R \! \begin{bmatrix}
			r_1 \\ r_2 \\ \vdots \\ r_K
		\end{bmatrix} \!  + \! H \! \begin{bmatrix}
			e_1 \\ e_2 \\ \vdots \\ e_p
		\end{bmatrix}\!,
\]}

\noindent where dependence on $q$ is omitted for compactness of notation. The compact form of this equation is given by
\begin{equation} \label{eq:network_model}
	w = G w + R r + He.
\end{equation}
This leads to the following definition of a network model \cite{Weerts&etal_Autom:18_identifiability}.
\begin{defn}[Network model] \label{de:model}
	A network model of $L$ nodes, and $K$ external excitation signals, with a noise process of rank $p\leq L$ is defined by the quadruple
		$M = (G, R, H, \Lambda)$,
	with
	\begin{itemize}
	   	\item $G\in \Real^{L\times L}(q)$, diagonal entries $0$, all modules strictly proper\footnote{The condition of strictly proper modules can be relaxed in relation with a possible diagonal structure of $\Lambda$ and the presence/absence of algebraic loops in the network, see \cite{Weerts&etal_Autom:18_identifiability}.} and stable;
		\item $R \in \Real^{L\times K} (q)$, proper, with in each row a single nonzero entry;
		\item $H \in \Real^{L\times p} (q)$, proper and stable, with a left stable inverse, and a $p\times p$ submatrix being monic;
		\item $\Lambda \in \Real^{p\times p}$, $\Lambda >0$;
		 \item ${(I-G)^{-1}}$ is proper and stable (well-posedness);
	\end{itemize}
\end{defn}
For the purpose of studying network identifiability we need the following definition of a network model set.
\begin{defn}[Network model set]
    A network model set $\Mcal$ is defined as
        $\Mcal = \{M(\theta) = (G(\theta),R,H(\theta),\Lambda(\theta)); \theta \in \Theta\}$,
    with $\Theta \subset \Real^{n_\theta}$ a finite-dimensional parameter space, where each model satisfies the conditions of Definition \ref{de:model}.
\end{defn}
In this model set $R$ is considered to be known and therefore fixed. The elements in $G$, $H$ and $\Lambda$ can be parametrized, but some of these elements also can be known and thus fixed.

For the concept of network identifiability we follow the concept introduced in \cite{Weerts&etal_Autom:18_identifiability} and an extension towards genericity in \cite{Hendrickx&Gevers&Bazanella_TAC:19}, that in the current setting revolves around the transfer function
    $T(q,\theta) = (I-G(q, \theta))^{-1} \begin{bmatrix} H(q,\theta) & R(q) \end{bmatrix}$
as follows.
\begin{defn}[Network identifiability]
For a network model set $\Mcal$, and a model $M(q,\theta_0) \in \Mcal$ we consider the implication
     \begin{equation} \label{eq:identifiability}
        T(q,\theta_0) = T(q,\theta_1) \implies M(\theta_0) = M(\theta_1)
    \end{equation}
    for all $\theta_1 \in \Theta$. Then $\Mcal$ is globally (generically) identifiable from $(r,w)$ if implication (\ref{eq:identifiability}) holds for all (almost all) $\theta_0 \in \Theta$.
\end{defn}
The notion of generic identifiability is particularly attractive as it allows to be tested on the basis of graph-based tests, see e.g. \cite{Hendrickx&Gevers&Bazanella_TAC:19,Cheng&etal_TAC:22,Shi&etal_Autom:22}. This will be summarized in the next Subsection.

\subsection{Graph representation}
The dynamic network interconnection structure can be represented by a directed graph $\breve{\Gcal}$ that consists of the finite set of vertices $\breve{\Vcal} := \{1,2, \dots, L \}$ and the edge set $\breve{\Ecal} :=  \{ (i,j) \in  \breve{\Vcal} \times \breve{\Vcal} \ | \ G_{ji} \neq 0 \}$. See for details of graph theory e.g., \cite{Mesbahi2010}. The graph is simple, since no self-loops are present in the network model. The correlation structure of the noise signals $v$ will also be included in this graph representation by defining an extended graph as follows.
\begin{defn}[Extended graph \cite{Cheng&etal_TAC:22}]
    Consider a directed dynamic network \eqref{eq:network_model}. Let ${\breve{\Gcal}=(\breve{\Vcal},\breve{\Ecal})}$ be its underlying graph. The extended graph ${\Gcal=(\Vcal,\Ecal)}$ is defined as 
	\begin{align}
		\Vcal &:= \breve{\Vcal} \cup \{L+1, L+2, \dots, L + p \} \\
		\Ecal &:= \breve{\Ecal} \cup \{ (i,j) \in \Vcal \times \breve{\Vcal} \ | \  H_{j,i-L}(q) \neq 0, i > L \},
	\end{align}
	with $L = |\breve{\Vcal}|$ the number of nodes and $p$ the number of noise signals $e(t)$.
\end{defn}
Actually the white noise sources $e$ have been added as nodes in the graph, and edges can appear from $e$-nodes to $w$-nodes, but not reversed.\\
We will refer to the sets $\Ncal_j^- := \{ i \in \Vcal \ | \ (i,j) \in \Ecal \}$ and  $\Ncal_j^+ := \{ i \in \Vcal \ | \ (j,i) \in \Ecal \}$ as the set of in- and out-neighbors of node $j$, respectively. A path that connects the vertices $i_0$ to $i_n$ is a sequence of edges of the form ${(i_{k-1}, i_k)}, {k=1,\dots,n}$, where every vertex occurs at most once on the path. A single vertex is also considered to have a path to itself. Two paths are vertex-disjoint if they do not share any nodes, including starting and ending nodes. The maximum number of vertex-disjoint paths from a set $\Acal \subseteq \Vcal$ to a set $\Bcal \subseteq \Vcal$ is denoted by the operator $b_{\Acal \to \Bcal}$.

Based on this network model set we define a graph representation of the model set, where both the parametrized and fixed modules and elements in the $H$ matrix are explicitly represented by a set of parametrized $\Ecal_p$ and fixed edges $\Ecal_f$.
\begin{defn}[Graph representation of network model set]
A network model set $\Mcal$ has a graph representation through the extended graph $\Gcal = (\Vcal, \Ecal_p \cup \Ecal_f)$, where $\Ecal_p \cup \Ecal_f = \Ecal$ and $\Ecal_p \cap \Ecal_f = \emptyset$ determined by
\begin{equation*}
\begin{aligned}
            \Ecal_p &:= \{ (i,j) \in \Ecal \ | \ \mbox{edge } (i,j) \text{ is parametrized}\}, \mbox{and} \\
            \Ecal_f &:= \{ (i,j) \in \Ecal \ | \ \mbox{edge } (i,j) \text{ is fixed}\}.
        \end{aligned}
\end{equation*}
%
%
\end{defn}

This graph representation can be used to establish conditions for generic network identifiability of the network model set $\Mcal$, given some assumptions regarding the parametrization. The following assumptions are used throughout the paper.
\begin{assumption}[\!\!\cite{Shi&etal_Autom:22}] \label{as:identifiability}  \hfill
\begin{enumerate}
    \item All the parametrized entries in $M(\theta)$ are parametrized independently.
    \item In model set $\Mcal$, the rank of any submatrix of $\begin{bmatrix}
        G(q,\theta)-I & H(q,\theta) & R(q)
    \end{bmatrix}$ that does not depend on $\theta$, is equal to its structural rank\footnote{The structural rank of a matrix is the maximum rank of all matrices with the same nonzero pattern \cite{Steffen2005}.}.
\end{enumerate}
\end{assumption}
The second assumption here ensures that the numerical values of the fixed entries in the model set do not induce any singularity in the considered matrix. This assumption allows to formulate  conditions for identifiabiltiy on the basis of the graph representation of the model set.
\begin{lemma}[Generic network identifiability] \label{lem:Identifiability}
Consider a network model set $\Mcal$ with graph representation $\Gcal$ that satisfies Assumption \ref{as:identifiability}, and let $\Ucal := \Rcal \cup \{L+1,\cdots,L+p\}$ with $\Rcal \subset \breve\Vcal$ the set of $w$-vertices that are directly excited by $r$ signals. Then, the network model set $\Mcal$ is generically identifiable from $(w,r)$ if in $\Gcal$ it holds that
	$b_{\Ucal \to \Pcal_j} = |\Pcal_j|$
	for all nodes $j \in \breve{\Vcal}$, where
\begin{equation} \label{eq:Pj}
\Pcal_j := \{ i \in \Ncal_j^{-} \ | \ (i,j) \in \Ecal_p \}.
\end{equation}
\end{lemma}
\begin{proof}
This Lemma is formulated and proven as Lemma 2 in \cite{Cheng&etal_TAC:22} but for the set $\Pcal_j$ replaced by the set of all in-neighbours of $w_j$ in the extended graph. With the subsystems generic identifiability result from Theorem 3 in \cite{Shi&etal_Autom:22} it can simply be shown that the original Lemma extends to the situation of considering only the subset $\Pcal_j$ of parametrized modules, rather than the full set of in-neighbours of node $w_j$.
\end{proof}


Whereas Lemma \ref{lem:Identifiability} provides the condition for {\it verifying} identifiability node-wise, the target is to devise a method to allocate excitation signals that guarantees that the conditions of Lemma \ref{lem:Identifiability} are satisfied for all nodes simultaneously.

\section{Motivating example \& Problem statement}
\label{sec:motivexam}
The original formulation of Lemma \ref{lem:Identifiability} in \cite{Cheng&etal_TAC:22} has led to an algorithm for the allocation of external excitation signals based on a covering of the extended graph of the model set with disjoint pseudotrees. A pseudotree is a subgraph in which the maximum indegree of each vertex is $1$. This implies that a pseudotree can be built up of a cycle with outgoing trees starting from the vertices in the cycle. Pseudotrees are disjoint if they do not share any edges, and all outgoing edges of each vertex belong to the same pseudotree \cite{Cheng&etal_TAC:22}.

An illustrative example of a pseudotree covering is provided in Figure \ref{fig:MotivatingExample}\subref{fig:MotEx1} where all modules are parametrized, and the graph is covered by two disjoint pseudotrees indicated in red and blue. Generic network identifiability is guaranteed if one of the roots in each pseudotree is excited, where in the case of a pseudotree with a cycle each vertex in the cycle acts as a root. In the example in Figure \ref{fig:MotivatingExample}\subref{fig:MotEx1} this implies that excitation at node $w_2$ together with one of the nodes $(w_1, w_4, w_5)$ is sufficient for guaranteeing generic identifiability. The result in \cite{Cheng&etal_TAC:22} extends to the situation where some modules are fixed, by requiring that only the subgraph composed of parametrized modules should be covered in the pseudotree covering.
An example of this is provided in
Figure \ref{fig:MotivatingExample}\subref{fig:MotEx2}, where the fixed modules are represented by dashed lines.
If we remove the fixed edges and cover the remaining graph with pseudotrees, excitation signals will be allocated at $w_4$ and $w_5$, showing that the number of excitations increases when some modules are fixed, being a result of conservatism in the algorithm.
At the same time it can be observed that the single pseudotree
$w_5 - w_1 - w_4 - w_2 - w_3 - w_5$, covers all parametrized edges, and so would lead to allocation of one external excitation at any of the nodes in this cycle. However this solution, with a reduced number of excitations, will typically not be found by the pseudotree-covering algorithm. In this article we will extend the allocation algorithm of \cite{Cheng&etal_TAC:22} to effectively deal with the fixed modules in a structural way.
%
%
%


\begin{figure}%
\centering
\subfloat[]{\includegraphics[scale=0.8]{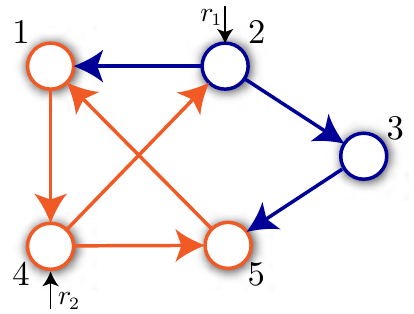}\label{fig:MotEx1}}
\subfloat[]{\includegraphics[scale=0.8]{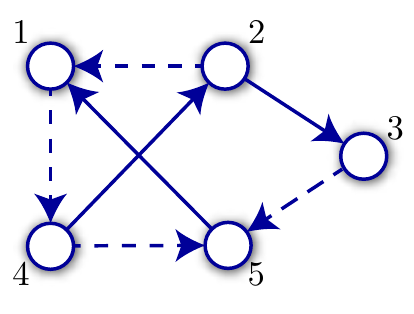}\label{fig:MotEx2}}
\caption[]{An example of a network model set with: (\subref*{fig:MotEx1}) only parametrized modules (solid) and (\subref*{fig:MotEx2}) including fixed modules (dashed).}
\label{fig:MotivatingExample}
\vspace{-.4cm}
\end{figure}


\section{Graph covering} \label{sec:approach}
%
The graphical approach to allocate excitation signals comes down to decomposing the network in sub-graphs with the property that in each sub-graph all of its modules are identifiable by applying a single excitation source. Given this decomposition, it is shown that applying excitation signals to a specific set of nodes associated with these sub-graphs leads to an identifiable network model set.

The sub-graphs that are used belong to a family of graphs, called the multi-rooted graphs, defined as follows.
\begin{defn}[Multi-rooted graph]
    A connected directed graph $\Tcal$, with $|V(\Tcal)| \geq 2$, is called a (directed) \textbf{multi-rooted graph} if there exists a path from each node in a nonempty set of \textbf{roots} $\Upsilon(\Tcal)$ to every node $i \in V(\Tcal)$.
\end{defn}
%
For our identifiability study we will restrict attention to a particular class of multi-rooted graphs, for which each node only has maximally one incoming parametrized edge.
\begin{defn}[Single-source identifiable multi-rooted graph]
    A multi-rooted graph $\Tcal$ is called a \textbf{single-source identifiable multi-rooted graph} (SIMUG) if all nodes $j \in V(\Tcal)$ satisfy ${|\Pcal_j|\leq1}$, with $\Pcal_j$ defined in (\ref{eq:Pj}).
\end{defn}
\begin{figure}[b]
\vspace{-.4cm}
\centering\subfloat[]{\includegraphics[width=0.25\columnwidth]{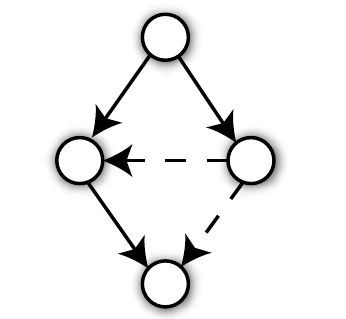}}
\subfloat[]{\includegraphics[width=0.25\columnwidth]{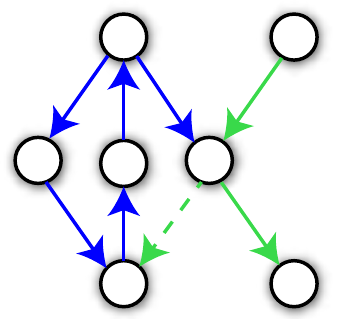}}
\subfloat[]{\includegraphics[width=0.25\columnwidth]{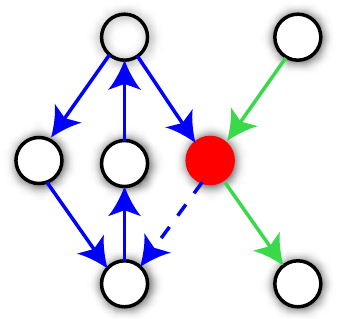}}
\caption{(a) Example of a SIMUG, the maximum number of parametrized incoming edges of a node does not exceed $1$; (b) a covering with two edge-disjoint SIMUGs; (c)
The two SIMUGs are not edge-disjoint since one (red-colored) node has two outgoing edges from different SIMUGs. Fixed edges are dashed arrows.}
\label{fig:SIMUG}
\end{figure}
\begin{defn}[Edge-disjoint SIMUGs, \cite{Cheng&etal_TAC:22}] \label{de:disjoint}
    Consider two multi-rooted graphs $\Tcal_1$ and $\Tcal_2$ as subgraphs of a directed graph $\Gcal$. $\Tcal_1$ and $\Tcal_2$ are called \textbf{edge-disjoint} in $\Gcal$ if the following two conditions hold:
    \begin{enumerate}
        \item $\Tcal_1$ and $\Tcal_2$ have no common edges;
        \item For each node $j$ in $\Tcal_1 \cup \Tcal_2$, all outgoing edges in $\Tcal_1 \cup \Tcal_2$ belong to either $\Tcal_1$ or to $\Tcal_2$.
    \end{enumerate}
\end{defn}
%
%
The concept of edge-disjoint SIMUGs is illustrated in Figure \ref{fig:SIMUG}.
In line with the disjoint edge-covering used in \cite{Cheng&etal_TAC:22} we can now define:
\begin{defn}[Edge-disjoint SIMUG covering]
   Consider a directed graph $\Gcal$, and let $\Pi := \{ \Tcal_1, \Tcal_2, \dots, \Tcal_n \}$ be a collection of edge-disjoint SIMUGs of $\Gcal$. The edges in a set $\Ecal \subseteq E(\Gcal)$ are covered by $\Pi$, if $E(\Tcal_1) \cup E(\Tcal_2) \cup \dots \cup E(\Tcal_n) = \Ecal$, and $\Pi$ is called an edge-disjoint SIMUG covering of $\Ecal$.
\end{defn}

It follows directly from the existence results of pseudotree coverings in \cite{Cheng&etal_TAC:22} that every graph can be covered by edge-disjoint SIMUGs.
The concepts that are defined above provide the means for specifying the conditions for allocating external excitaton signals so as to guarantee network identifiability.


\begin{theorem}\label{th:MRG_ident}
    Consider a network model set $\Mcal$ that satisfies Assumption \ref{as:identifiability} and let $\Gcal = (\Vcal,\Ecal)$ be its extended graph.
    Then $\Mcal$ is generically identifiable from $(r,w)$ if there exists an edge-disjoint SIMUG covering of $\Ecal$, denoted by $\Pi = \{\Tcal_1, \Tcal_2, \dots, \Tcal_{n}\}$, such that $\forall k \in \{1,2, \dots, n\}$ 
    there exists a vertex $\tau_k \in \Upsilon(\Tcal_k)$ that is externally excited by an independent $r$ or $e$-signal.
\end{theorem}
\begin{proof}
The proof is analogous to \cite[Th.~1]{Cheng&etal_TAC:22}, with the difference that the pseudotrees that were used in
\cite{Cheng&etal_TAC:22} have to be replaced by the SIMUGs introduced here.
%
The following reasoning from \cite[Th.~1]{Cheng&etal_TAC:22}, modified to the setting in this paper, still applies: ``The disjointness of the SIMUGs in Definition \ref{de:disjoint} implies that the paths in different disjoint SIMUGs are vertex-disjoint, if they have no common starting or ending nodes, and, for any vertex $j \in V(\Gcal)$, all the edges incident from the vertices in $\Pcal_j$ to $j$ should belong to distinct SIMUGs. Furthermore, any two disjoint SIMUGs cannot share common root nodes, and thus $\tau_i \neq \tau_j$, for all $i \neq j$. Consequently, the above properties of disjoint SIMUGs yield that there exist $|\Pcal_j|$ vertex-disjoint paths from $\{\tau_1, \tau_2,\dots, \tau_{n}\}$ to $\Pcal_j$." Then, application of Lemma \ref{lem:Identifiability} shows the result.
\end{proof}
With the result of Theorem \ref{th:MRG_ident}, we have reduced the allocation problem to the problem of finding a minimal number of SIMUGs that cover the extended graph.

\section{Allocation algorithm}
\label{sec:allocation}
\subsection{Introduction}
For formulating an algorithm that generates a SIMUG covering of the network, we follow an approach which is very closely related to the algorithm in \cite{Cheng&etal_TAC:22}, however adapted to the situation of having SIMUGs with possibly fixed modules. As a strategy we are going to start with an initial SIMUG covering of the network, and then we are going to merge SIMUGs so as to arrive at a smaller number of SIMUGs that cover the network. In this setting we have to formulate appropriate conditions for merging two SIMUGs, and we need new algebra for formulating the merging algorithm.

As an initial covering we start with the SIMUG covering
	$\Pi_0 = \{\Tcal_1^{(0)}, \Tcal_2^{(0)}, \dots, \Tcal_{|\Pi_0|}^{(0)}\}$
where $|\Pi_0| = |\Vcal| - |\Scal_{in}|$ with $\Scal_{in}$ the set of sinks in the graph, where for each node $k \in \Vcal$ that is not a sink, $\Tcal_k^{(0)}$ is composed of node $k$ and all of its outgoing edges in $\Ecal$.
%
%
\begin{defn}[Mergeability \cite{Cheng&etal_TAC:22}] \label{de:Mergeability}
		Consider two disjoint SIMUGs $\Tcal_1$ and $\Tcal_2$ and $V(\Tcal_1) \cap V(\Tcal_2) = \emptyset$. We say that $\Tcal_1$ is mergeable to $\Tcal_2$, if
		\begin{enumerate}
			\item the union of $\Tcal_1$ and $\Tcal_2$, i.e., $(V(\Tcal_1) \cup V(\Tcal_2) , E(\Tcal_1) \cup E(\Tcal_2))$ is also a SIMUG;
			\item and there is a directed path from every vertex $i \in \Upsilon(\Tcal_2)$ to every vertex $j\in V(\Tcal_1)$. 
		\end{enumerate}
	\end{defn}

\medskip
To every SIMUG covering of the network a characteristic matrix is connected, that will be used for steering the merging operations of the SIMUGs. It is defined as follows.

\begin{defn} \label{de:MergMat}
	Denote a set $\bS =\{1,0, \varnothing \}$ and
	let $\Pi = \{ \Tcal_1, \Tcal_2, \dots, \Tcal_n \}$ be a disjoint SIMUG covering of a directed graph. The characteristic matrix of $\Pi$ is denoted by $\fM \in \bS^{n\times n}$, whose $(i,j)$-$th$ entry is defined as
	\begin{equation} \label{eq:charaM}
		\fM_{ij} = \begin{cases}
			1 & \text{if } \Tcal_i \text{ is mergeable to } \Tcal_j \\
			\varnothing & \text{if } \Zcal_{ij} = \emptyset, \\
			0 & \text{otherwise}
		\end{cases}
	\end{equation}
	where $\Zcal_{ij} :=
    \{ x \in V(\Tcal_j) \cup V(\Tcal_i) \ | \ \lvert \Pcal_x \cap (V(\Tcal_j) \cup V(\Tcal_i)) \rvert > 1 \}$
    is the set of nodes in $\Tcal_i$ and $\Tcal_j$ that have multiple parametrized in-neighbors in the two SIMUGs. 
\end{defn}

The characteristic matrix in \eqref{eq:charaM} will serve as an algebraic means for merging SIMUGs in an algorithm that provides a covering of the  network with a reduced number of SIMUGs.

\subsection{Initial/characteristic matrix algorithm}
For specification of the characteristic matrix of the initial covering we follow a reasoning similar to the result of Lemma 5 in \cite{Cheng&etal_TAC:22}. However, due to the decomposition of $\Gcal$ in $\Gcal_f$ and $\Gcal_p$, this reasoning leads to new technical expressions that require an independent proof.

%
\begin{lemma} \label{lemma2} Given the extended graph representation ${\Gcal = (\Vcal, \Ecal_p \cup \Ecal_f)}$ of the network model set $\Mcal$, for which the edges are separated in the parametrized $\Gcal_p = (\Vcal, \Ecal_{p})$ and fixed graph $\Gcal_f = (\Vcal, \Ecal_{f})$ with $n = |\Vcal|$, the corresponding adjacency matrices are given by $A(\Gcal_p)$ and $A(\Gcal_f)$, respectively. Denote
	\begin{equation} \label{eq:aij}
		a_{ij} = \left( [n A(\Gcal_p)-\!A(\Gcal_f) +\!I\mathrm{i}]_{\star i} \right)^{\!\top} \left( [n A(\Gcal_p) -\!A(\Gcal_f)]_{\star j} \right)\!
	\end{equation}
	where $i,j \in 1,2,\dots,|\Pi_0|$, $\mathrm{i}$ denotes the imaginary unit.
Then, the characteristic matrix $\fM^{(0)}$ of $\Pi_0$ is given by: $\fM^{(0)}_{ii}=0$ for all $i$, while for $j\neq i$:
	\begin{equation} \label{eq:CharacteristicMatrixAlgorithm}
		\fM_{ij}^{(0)} = \begin{cases}
			1 & \mbox{if }\mathrm{Re}(a_{ij}) < n \text{ and } \mathrm{Im}(a_{ij}) \neq 0; \\
			\varnothing & \mbox{if }\mathrm{Re}(a_{ij}) < n\text{ and } \mathrm{Im}(a_{ij}) = 0; \\
			0 & \mbox{if }\mathrm{Re}(a_{ij}) \geq n.
		\end{cases}
	\end{equation}
	where $\mathrm{Re}(\cdot)$ and $\mathrm{Im}(\cdot)$ denote the real and imaginary parts of a complex number.
\end{lemma}

\begin{proof}
This proof essentially shows that the initial mergeability matrix in Definition \ref{de:MergMat} can be found, using \eqref{eq:aij} based on the adjacency matrices of the fixed and parametrized parts of the extended graph.

\textbf{\textit{Case $\mathrm{Re}(a_{ij}) \geq n$:}} This result indicates that multiple parametrized edges enter the same node from both SIMUGs, which results in a nonmergeable `$0$' entry in the matrix. 
The distinction between parametrized and fixed edges is made due to the scaling factor $n$, which
leads to
\begin{equation} \label{eq:LemPrGeq}
    \mathrm{Re}(a_{ij}) =  x_1 n^2 - x_2 n + x_3 \geq n,
\end{equation}
with $x_1, x_2, x_3 \in \Nat$ denoting the occurrences of respectively common parametrized edges to the same node, i.e., $x_1 = |\Zcal_{ij}|$, a parametrized and fixed edge to the same node or common fixed nodes. Note that due to the number of nodes the following constraint holds
\begin{equation} \label{eq:LemConstraint}
    x_1 + x_2 + x_3 \leq n - 2,
\end{equation}
since $[n A(\Gcal_p) - A(\Gcal_f)]_{i i} = 0$ by exclusion of self-loops. Then, we show that $\mathrm{Re}(a_{ij}) \geq n$ if and only if $x_1 > 0$, since for $x_1 = 0$ we have $\mathrm{max}_{x_2,x_3} - x_2 n + x_3 = n-2$ with $x_2 = 0$ and $x_3 = n-2$ by \eqref{eq:LemConstraint}. Now, also any solution for $x_1>0$ implies that \eqref{eq:LemPrGeq} holds, since $x_2$ is constrained by $x_2 \leq n-2 - x_1$. Therefore, $\mathrm{Re}(a_{ij}) \geq n$ results in $\fM_{ij} = 0$, because the union of $\Tcal_i$ and $\Tcal_j$ is clearly not a SIMUG since a node has $|\Pcal_j| > 1$.

\textbf{\textit{Case $\mathrm{Re}(a_{ij}) \leq n$ \& $\mathrm{Im}(a_{ij}) \neq 0$:}} Here, $x_1 =0$ in \eqref{eq:LemPrGeq}, so the constraint $|\Pcal_j| \leq 1$ holds for the union of $\Tcal_i$ and $\Tcal_j$, while the imaginary axis is used to encode whether the 2-nd condition in Definition \ref{de:Mergeability} is satisfied. This ensures that in case of merging $\Tcal_i$ to $\Tcal_j$ the root(s) of $\Tcal_j$ still have a path to all nodes in $\Tcal_i$. The root node(s) in $\Tcal_i$ is/are encoded by $\mathrm{i}$ in \eqref{eq:aij}, which has a path to all its nodes. Then, if an edge exists from a node in $\Tcal_j$ to the root denoted by $\mathrm{i}$ in $\Tcal_i$, there also exists a path from the root of $\Tcal_j$ to all the nodes in $\Tcal_i$, so $\fM_{ij} = 1$.

\textbf{\textit{Case $\mathrm{Re}(a_{ij}) \leq n$ \& $\mathrm{Im}(a_{ij}) = 0$:}} In this case the same reasoning holds, but the result is not mergeable, since there is not a path from the root(s) of $\Tcal_j$ to the nodes in $\Tcal_i$. Therefore, the result is $\fM_{ij} = \varnothing$, which concludes the proof.
\end{proof}

With Lemma \ref{lemma2} the characteristic matrix of the initial covering can directly be calculated on the basis of the adjacency matrices $A(\Gcal_p)$ and $A(\Gcal_f)$.

\subsection{Merging procedure}
Next we will represent a merging operation on a SIMUG covering through an equivalent operation of the characteristic matrix.
In view of this we define $\fM \in \bS^{|\Pi_0|\times |\Pi_0|}$, and let $\fM_{i\star}$ and $\fM_{\star j}$ be the $i$-th row and $j$-th column of   $\fM$, respectively.

For a given disjoint SIMUG covering $\Pi$ with $|\Pi| = n$ and a set $\mathbb{N} := \{ 1,2, \dots, n \}$, we then define the following function
	$\mathcal{F} : \bS^{n\times n} \times \mathbb{N} \times \mathbb{N} \to \bS^{(n-1) \times (n-1)}$,
and $\hat{\fM} = \mathcal{F}(\fM, i,j)$ is a reduction of $\fM$ obtained by merging SIMUG $i$ into SIMUG $j$ by the following algebraic operations:
1) $\hat{\fM} = \fM$;
2) Row merging: $\hat{\fM}_{j \star} = \fM_{i\star} \otimes \fM_{j \star}$;
3) Column merging: $\hat{\fM}_{\star j} = \fM_{\star i} \odot \fM_{\star j}$;
4) Remove the $i$-th row and column of $\hat{\fM}$.

The row merging operator $\otimes$ and the column operator $\odot$ are given as follows. First, the column operator $\odot$ is commutative defined to describe the merging feature algebraically:
	$c = a \odot b = b \odot a$,
with $a,b,c \in \bS$, which follow the rules
\begin{align}
	& 1 \odot 1 = 1, 1 \odot 0 = 0, 1 \odot \varnothing = 1, \label{eq:col_merge_rules1} \\
	& 0 \odot 0 = 0, \varnothing \odot 0 = 0, \varnothing \odot \varnothing = \varnothing, \label{eq:col_merge_rules2}
\end{align}
Then, the row merging operator $\otimes$ is similar and
adheres to the same rules as \eqref{eq:col_merge_rules1}-\eqref{eq:col_merge_rules2}, but differs in the rule $1 \odot \varnothing = 1$, where it is also not commutative. Instead of $1 \odot \varnothing = 1$ the following rules apply:
\begin{equation} \label{eq:row_merge_rules}
		\varnothing \otimes 1 = 1, 1 \otimes \varnothing = \varnothing.
\end{equation}

Now, we extend the operators to entrywise vector operations, for which we let $\rho, \mu \in \bS^n$ be two column (or row) vectors. The following operations $\rho \odot \mu$ and $\rho \otimes \mu$ are entrywise operations that return a new column (or row) vector, whose $i$-th element is given by $\rho_i \odot \mu_i$ or $\rho_i \otimes \mu_i$, respectively.

\begin{theorem}
	Consider a directed graph $\hat{\Gcal}$, and let $\Pi$ be a disjoint SIMUG covering of $\hat{\Gcal}$ represented by the characteristic matrix $\fM$. If in $\Pi$, the $i$-th SIMUG is mergeable to the $j$-th one, then a new SIMUG covering $\hat{\Pi}$ of $\hat{\Gcal}$ is obtained by merging the $i$-th SIMUG into the $j$-th one, where $|\hat{\Pi}| = |\Pi| - 1$ and the characteristic matrix of $\hat{\Pi}$ is given by $\hat{\fM}= \mathcal{F}(\fM,i,j)$.
\end{theorem}

\begin{proof}
This proof shows that, given a SIMUG covering, merging two SIMUGs and finding subsequent mergeable SIMUGs for the new covering is equivalent to applying the rules in \eqref{eq:col_merge_rules1}-\eqref{eq:row_merge_rules} to the mergeability matrix. First, the rules for the column merging is treated, followed by the row merging operator. Suppose $\Tcal_i$ is mergeable to $\Tcal_j$, then, the row operator defines for the reduced covering $\hat{\Pi}$ to which SIMUGs the union of $\Tcal_i$ and $\Tcal_j$ is mergeable, while the column operator defines which SIMUGs are mergeable to the union of $\Tcal_i$ and $\Tcal_j$. Therefore, the following statements hold for the column operator $\odot$ and the row operator $\otimes$ in \eqref{eq:col_merge_rules1}-\eqref{eq:col_merge_rules2}:
\begin{enumerate}
    \item If either $\Tcal_i$ or $\Tcal_j$ can not merge to or be merged to $\Tcal_x$ due to a '$0$' entry in the mergeability matrix, then, $\Tcal_x$ can not merge to or be merged to the union of $\Tcal_i$ and $\Tcal_j$. Hence, the rules $0 \odot 0 = 0$, $\varnothing \odot 0 = 0$ and $1 \odot 0 = 0$. \label{test}

    \vspace*{-4mm}
    \item If $\Tcal_x$ is mergeable to both $\Tcal_i$ and $\Tcal_j$, then, $\Tcal_x$ is also mergeable to the union of $\Tcal_i$ and $\Tcal_j$ and vice versa, i.e., $1\odot 1 = 1$.
    \item If $\fM_{xi} = \varnothing$, then, $\Tcal_x$ is not mergeable to $\Tcal_i$, but if $\Tcal_x$ is mergeable to $\Tcal_j$, then, $\Tcal_x$ is also mergeable to the union of $\Tcal_i$ and $\Tcal_j$ or vice versa, since the condition for a SIMUG that $|\Pcal_l| \leq 1 \ \forall \ l \in V(\Tcal_i \cup \Tcal_j \cup \Tcal_x)$ will still be satisfied. i.e., $\varnothing \odot 1 = 1$. Moreover, if $\fM_{xj} = \varnothing$ too, then, the same condition is satisfied, so $\varnothing \odot \varnothing = \varnothing$.
\end{enumerate}
Then, the row operator $\otimes$ has one logical difference with respect to the column operator $\odot$, i.e. \eqref{eq:row_merge_rules}, which is shown in the following statements:
\begin{enumerate}
    \item For the rule: $\varnothing \otimes 1 = 1$, for which $\fM_{ix} = \varnothing$ and $\fM_{jx} = 1$, $\Tcal_i$ is merged to $\Tcal_j$, so the root(s) of the union of $\Tcal_i$ and $\Tcal_j$ is/are the root(s) of $\Tcal_j$. Since $\Tcal_j$ was mergeable to $\Tcal_x$ and $\Tcal_i$ does not pose a problem to merge to $\Tcal_x$, the union of $\Tcal_i$ and $\Tcal_j$ is mergeable to $\Tcal_x$.
    \item Now for the rule: $1 \otimes \varnothing = \varnothing$ for which $\fM_{ix} = 1$ and $\fM_{jx} = \varnothing$, $\Tcal_i$ is merged to $\Tcal_j$, so the root(s) of the union of $\Tcal_i$ and $\Tcal_j$ is/are the root(s) of $\Tcal_j$. However, the root(s) of the union of $\Tcal_i$ and $\Tcal_j$ do not necessarily have a path to all nodes in $\Tcal_x$.
\end{enumerate}
Following these statements, the operators that are used in the function $\Fcal(\fM,i,j)$ are shown to be equivalent to merging two SIMUGs in the covering. The function $\Fcal(\fM,i,j)$ therefore reduces the mergeability matrix and updates it according to the merged SIMUGs.
\end{proof}

On the basis of the above merging operation we can now follow the same merging algorithm as presented in \cite{Cheng&etal_TAC:22}, where we first find the row of the characteristic matrix with a unique `$1$' entry, on the basis of which the SIMUG related to this row number is merged. Then we continue with merging SIMUGs related to `$1$' entries in nodes, giving priority to those rows that have the most `$\varempty$' entries. For more details and motivations we refer to Algorithm 1 in \cite{Cheng&etal_TAC:22}. Note that, like in \cite{Cheng&etal_TAC:22} there is no formal guarantee that we arrive at a minimal number of SIMUGs covering the network.

\section{Allocating excitation signals}
\label{sec:allo}

With the result of Theorem \ref{th:MRG_ident} we can guarantee generic network identifiability if we make sure that in an edge-disjoint SIMUG covering of the network graph, one root in every SIMUG is excited by an external $r$ or $e$-signal. Having such a SIMUG covering obtained in the previous Section, we can now detect those SIMUGs that do not have such an excitation yet from existing $r$ and $e$ signals, and allocate  additional $r$ signals to the root nodes of these SIMUGs.

In comparison with the situation of all-parametrized modules \cite{Cheng&etal_TAC:22}, there are two distinctive situations to be mentioned:
\begin{itemize}
\item For a SIMUG for which all roots have only outgoing links that are fixed, the external excitation does not necessarily need to be added to the root.
\item SIMUGs that are composed of fixed edges only, do not require excitation, and they do not need to be merged to other nodes. However they can serve in the role of interconnecting two other SIMUGs that do require excitation, and for which merging with the fixed SIMUG can reduce the number of excitation signals required.
\end{itemize}

After having allocated a sufficient number of $r$ signals conditions, a final check is made for necessity of all allocated $r$ signals, by individually removing them and verifying the vertex disjoint path condition of Lemma \ref{lem:Identifiability}.

\section{Example}
The presented algorithm is applied to an $8$-node network depicted in Figure \ref{fig:AllocationExample} that shows typical locations of fixed modules that introduces conservatism in the methods in \cite{Cheng&etal_TAC:22}.
For generic network identifiability external (excitation) signals are added to a root node of each SIMUG, in this case, one of the nodes $\{w_1,w_6\}$, and one of the nodes $\{w_7,w_8\}$.

\begin{figure}[htp]
\subfloat[]{\includegraphics[width=0.48\columnwidth]{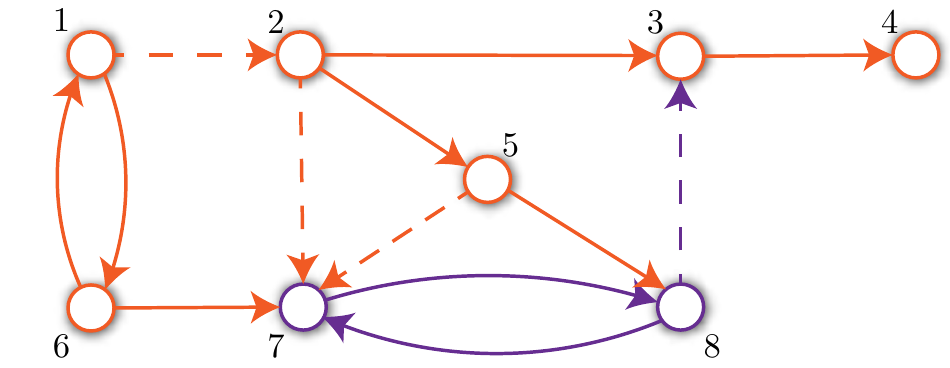}}
\subfloat[]{\includegraphics[width=0.48\columnwidth]{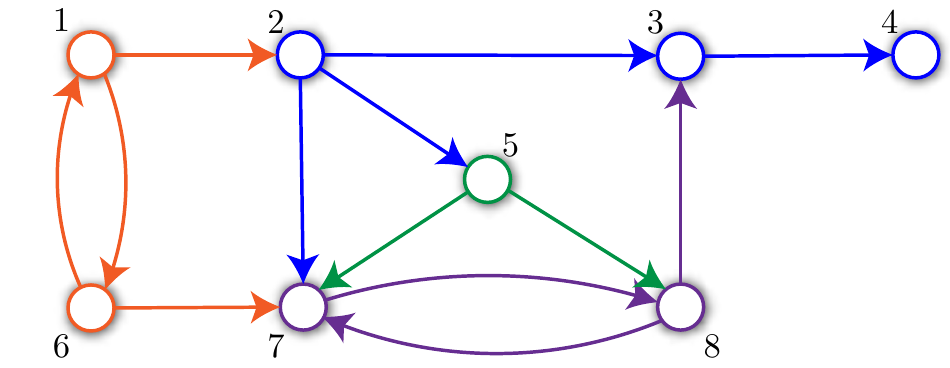}}
\caption{(a) A disjoint SIMUG covering of a network with $2$ SIMUGs; fixed modules are indicated with dashed lines, parametrized modules with solid lines. (b) A pseudotree covering of the same network with all modules parametrized.}
\label{fig:AllocationExample}
\vspace{-.6cm}
\end{figure}


This result can be compared against the pseudotree covering in \cite{Cheng&etal_TAC:22} in two ways. First, the approach where the fixed edges are excluded to cover the parametrized modules only does not consider the connection between node $w_1$ and $w_2$. Hence, another pseudotree is required with node $w_2$ as root. Second, if in this network all present modules would be parametrized, the resulting pseudotree covering is depicted in Figure \ref{fig:AllocationExample}(b). In this situation we arrive at four pseudotrees leading to additional (excitation) signals to be added at nodes $w_2$ and $w_5$.




\section{Conclusion}
A graphical method for allocating external signals to achieve generic identifiability of a dynamic network model set, has been extended to the situation where network modules can be fixed. To this end, a pseudotree covering of the network graph has been generalized to a covering based on a special type of multi-rooted graph (SIMUG). The related merging algorithm also has been generalized, aiming at the allocation of a minimum number of external excitation signals for achieving generic identifiability. While minimality of the number of allocated signals cannot be guaranteed, including the property that some modules might be known, and thus fixed, reduces the number of external signals that guarantee identifiability.

\bibliographystyle{IEEEtran}
\bibliography{literature,Paul_Dynamic_Networks_Library}

\begin{thebibliography}{10}
\providecommand{\url}[1]{#1}
\csname url@samestyle\endcsname
\providecommand{\newblock}{\relax}
\providecommand{\bibinfo}[2]{#2}
\providecommand{\BIBentrySTDinterwordspacing}{\spaceskip=0pt\relax}
\providecommand{\BIBentryALTinterwordstretchfactor}{4}
\providecommand{\BIBentryALTinterwordspacing}{\spaceskip=\fontdimen2\font plus
\BIBentryALTinterwordstretchfactor\fontdimen3\font minus
  \fontdimen4\font\relax}
\providecommand{\BIBforeignlanguage}[2]{{%
\expandafter\ifx\csname l@#1\endcsname\relax
\typeout{** WARNING: IEEEtran.bst: No hyphenation pattern has been}%
\typeout{** loaded for the language `#1'. Using the pattern for}%
\typeout{** the default language instead.}%
\else
\language=\csname l@#1\endcsname
\fi
#2}}
\providecommand{\BIBdecl}{\relax}
\BIBdecl

\bibitem{Goncalves&Warnick:08}
J.~Gon\c{c}alves and S.~Warnick, ``Necessary and sufficient conditions for
  dynamical structure reconstruction of {LTI} networks,'' \emph{IEEE Trans.
  Automatic Control}, vol.~53, no.~7, pp. 1670--1674, Aug. 2008.

\bibitem{Materassi&Innocenti:10}
D.~Materassi and G.~Innocenti, ``Topological identification in networks of
  dynamical systems,'' \emph{IEEE Trans. Automatic Control}, vol.~55, no.~8,
  pp. 1860--1871, 2010.

\bibitem{Nabi&Mesbahi:12}
M.~Nabi-Abdolyousefi and M.~Mesbahi, ``Network identification via node
  knockout,'' \emph{IEEE Trans. Automatic Control}, vol.~57, no.~12, pp.
  3214--3219, December 2012.

\bibitem{VandenHof&etal_Autom:13}
P.~M.~J. Van~den Hof, A.~G. Dankers, P.~S.~C. Heuberger, and X.~Bombois,
  ``Identification of dynamic models in complex networks with prediction error
  methods - basic methods for consistent module estimates,'' \emph{Automatica},
  vol.~49, no.~10, pp. 2994--3006, 2013.

\bibitem{Materassi&Salapaka:20}
D.~Materassi and M.~V. Salapaka, ``Signal selection for estimation and
  identification in networks of dynamic systems: a graphical model approach,''
  \emph{IEEE Trans. Automatic Control}, vol.~65, no.~10, pp. 4138--4153,
  october 2020.

\bibitem{Ramaswamy&VandenHof_TAC:21}
K.~R. Ramaswamy and P.~M.~J. Van~den Hof, ``A local direct method for module
  identification in dynamic networks with correlated noise,'' \emph{IEEE Trans.
  Automatic Control}, vol.~66, pp. 5237--5252, November 2021.

\bibitem{Weerts&etal_Autom:18_identifiability}
H.~H.~M. Weerts, P.~M.~J. Van~den Hof, and A.~G. Dankers, ``Identifiability of
  linear dynamic networks,'' \emph{Automatica}, vol.~89, pp. 247--258, March
  2018.

\bibitem{Hendrickx&Gevers&Bazanella_TAC:19}
J.~Hendrickx, M.~Gevers, and A.~Bazanella, ``Identifiability of dynamical
  networks with partial node measurements,'' \emph{IEEE Trans. Autom. Control},
  vol.~64, no.~6, pp. 2240--2253, 2019.

\bibitem{vanWaarde&etal_TAC:20a}
H.~J. van Waarde, P.~Tesi, and M.~K. Camlibel, ``Necessary and sufficient
  topological conditions for identifiability of dynamical networks,''
  \emph{IEEE Trans. Autom. Control}, vol.~65, no.~11, pp. 4525--4537, november
  2020.

\bibitem{Shi&etal_Autom:22}
S.~Shi, X.~Cheng, and P.~M.~J. Van~den Hof, ``Generic identifiability of
  subnetworks in a linear dynamic network: the full measurement case,''
  \emph{Automatica}, vol. 137, no. 110093, March 2022.

\bibitem{Cheng&etal_TAC:22}
X.~Cheng, S.~Shi, and P.~M.~J. Van~den Hof, ``Allocation of excitation signals
  for generic identifiability of linear dynamic networks,'' \emph{IEEE Trans.
  Automatic Control}, vol.~67, no.~2, pp. 692--705, February 2022.

\bibitem{Mesbahi2010}
M.~Mesbahi and M.~Egerstedt, \emph{Graph Theoretic Methods in Multiagent
  Networks}.\hskip 1em plus 0.5em minus 0.4em\relax Princeton University Press,
  2010.

\bibitem{Steffen2005}
T.~Steffen, \emph{Control reconfiguration of dynamical systems: linear
  approaches and structural tests}.\hskip 1em plus 0.5em minus 0.4em\relax
  Springer Science \& Business Media, 2005, vol. 320.

\end{thebibliography}

\end{document}